%% file: main.tex
\documentclass[10pt,conference,letterpaper,table]{IEEEtran}

\usepackage{tikz}
\usetikzlibrary{shapes}
\usepackage{pgfplots}
\usepackage{filecontents}
\usepackage{float}
\usepackage{graphicx}
\usepackage{mdframed}
\usepackage{amssymb}
\usepackage{amsmath}
\usepackage{amsthm}
\usepackage{mathtools}
\usepackage{bbm}
\usepackage{hyperref}
\usepackage{cleveref}
\usepackage{booktabs}
\usepackage{caption}
\usepackage{subcaption}
\usepackage{float}
\theoremstyle{definition}

\newtheorem{claim}{Claim}
\newtheorem*{claim*}{Claim}

\newtheorem*{lemma*}{Lemma}
\newtheorem*{remark*}{Remark}

\newtheorem*{conjecture}{Conjecture}

\definecolor{chi}{rgb}{1,0,0}
\usepackage{cite}

\ifCLASSINFOpdf
\else
\fi
\IEEEoverridecommandlockouts
\begin{document}
%
\title{Age of Broadcast and Collection in Spatially Distributed Wireless Networks}
%
%
%

        \author{Chirag Rao and Eytan Modiano\thanks{Chirag Rao and Eytan Modiano are with the Laboratory for Information and Decision Systems (LIDS), Massachusetts Institute
of Technology, Cambridge, MA, 02139, USA. The final version of this paper will
appear in the proceedings of IEEE INFOCOM 2023.
E-mail: \{crao, modiano\}@mit.edu.}}

\maketitle
\input{texfiles/abstract}

%
\input{texfiles/intro}
\input{texfiles/model}

\input{texfiles/preliminaries}
\input{texfiles/broadcast}
\input{texfiles/collection}

\input{texfiles/numerical}
\input{texfiles/conclusion}

\input{texfiles/appendices}
\ifCLASSOPTIONcaptionsoff
  \newpage
\fi



\bibliographystyle{IEEEtran}
\bibliography{daoi,additional}


%







\end{document}

%% file: texfiles/abstract.tex
\begin{abstract}
  We consider a wireless network with a base station broadcasting and collecting time-sensitive data to and from spatially distributed nodes in the presence of wireless interference. The Age of Information (AoI) is the time that has elapsed since the most-recently delivered packet was generated, and captures the freshness of information. In the context of broadcast and collection, we define the Age of Broadcast (AoB) to be the amount of time elapsed until all nodes receive a fresh update, and the Age of Collection (AoC) as the amount of time that elapses until the base station receives an update from all nodes.
  We quantify the average broadcast and collection ages in two scenarios: 1) instance-dependent, in which the locations of all nodes and interferers are known, and 2) instance-independent, in which they are not known but are located randomly, and expected age is characterized with respect to node locations. In the instance-independent case, we show that AoB and AoC scale super-exponentially with respect to the radius of the region surrounding the base station. Simulation results highlight how expected AoB and AoC are affected by network parameters such as network density, medium access probability, and the size of the coverage region.
\end{abstract}

%% file: texfiles/intro.tex
\section{Introduction}
Collection and broadcast of fresh information over spatially-distributed wireless nodes is important for proper functioning of real-time systems, such as search-and-rescue drones or environmental monitoring using IoT sensors~\cite{9380899}. Dynamic environments and the lack of wired infrastructure necessitate deployment of highly-distributed, ad-hoc network of sensors to gather and send information updates wirelessly, where nodes must communicate with minimal coordination overhead using simple random access schemes.

Such networks must also operate under wireless communication constraints, including interference, fading, and path loss. Ensuring broadcast and collection of the freshest information possible in such a setting is a considerable challenge.

A popular paradigm for measuring the freshness of information observed from a process is the Age of Information (AoI)~\cite{sun2019age,kaul2012real,kaul2011minimizing}. The literature addressing AoI and wireless networks is extensive. Average and peak AoI in wireless networks were first characterized in~\cite{costa2014age}. Optimal wireless link scheduling was studied in~\cite{he2016optimal,he2016optimizing,hsu2017age,hsu2019scheduling}, relying on a centralized scheduler that is able to coordinate link activations, and the authors of~\cite{kadota2019minimizing} considered scheduling policies that minimize AoI in wireless networks with packets randomly arriving and queueing at the base station. In addition, the authors of~\cite{kadota2021age} studied scheduling with random packet arrivals in a random access setting. Several works have addressed AoI and broadcast. In particular, the authors of~\cite{kadota2016minimizing,kadota2018scheduling} found optimal centralized scheduling policies for broadcast from a base station to a number of nodes, minimizing functions of AoI such as Expected Weighted Sum AoI.  In~\cite{zhong_multicast_2017} the authors investigated AoI in multicast and broadcast networks with i.i.d. exponential (continuous-time) inter-packet delivery times. Works such as~\cite{kadota_scheduling_2018} investigated network scheduling to minimize AoI under general wireless channel unreliability, while~\cite{hsu2019scheduling} studied scheduling policies with random arrivals, modeling the problem as a Markov Decision Process. From an information theoretic perspective,~\cite{chen2019benefits} explored the effect of coding on the AoI in two-user broadcast networks, and ~\cite{wang2019broadcast} addressed AoI for Broadcast in CSMA/CS wireless networks, assuming network connectivity follows the Protocol Model~\cite{gupta_capacity_2000}. The authors of~\cite{farazi2018age} explored AoI in all-to-all broadcast wireless networks, deriving average and peak AoI using fundamental properties of graphs.


More recently, AoI in spatially-distributed networks has been investigated. The authors of~\cite{tripathi2021age} investigated data dissemination and gathering, modeling spatial separation as edges on a mobility graph. The authors of~\cite{mankar_spatial_2021,mankar_throughput_2021}, deployed stochastic geometry analysis to capture the spatiotemporal statistics of AoI in networks where nodes are distributed as a homogeneous point process. The authors of~\cite{yang_age_2020,yang_understanding_2020-1} optimized network parameters such as the medium access probability to minimize average and peak AoI, leveraging knowledge of the interference statistics of Poisson-distributed wireless networks. While AoI has been considered in spatially-distributed wireless networks, the important cases of wireless broadcast and collection in a spatially distributed network have not been addressed.

\textbf{Our main contribution in this work} is to introduce the notion of Age to the broadcast and collection of information.  We define two metrics -- the \textit{age of broadcast} (AoB) and the \textit{age of collection} (AoC) -- that characterize the amount of time elapsed since all receivers successfully receive an update in the broadcast case or all transmitters successfully deliver a packet to the base station in the collection case. We consider both the instance-dependent, and the instance-independent AoB and AoC. In the instance-dependent scenario, the locations of all interferers, transmitters and receivers are fixed and known. In the instance-independent scenario, the positioning of nodes and interferers is unknown but is distributed according to a Poisson point process. 



The rest of the paper is organized as follows. In~\Cref{sec:model}, we introduce the system model and define AoB and AoC. We then detail preliminaries in~\Cref{sec:preliminaries}. In~\Cref{sec:broadcast}, we characterize the expected AoB, then characterize AoC in~\Cref{sec:collection}. Numerical results from simulation are presented in~\Cref{sec:numerical}, and concluding remarks and future directions are stated in~\Cref{sec:conclusion}.

%% file: texfiles/model.tex
\section{System Model}\label{sec:model}
We now introduce the network model, the traffic model, as well as AoI before formally defining AoB and AoC. 

\textit{Notation}: Common notation can be found in~\Cref{tab:notation}. Whenever necessary for clarity, the expected value operator with respect to the distribution of some random element $X$ will be denoted by $\mathbb{E}_X[\cdot]$. The spatial point process models in this work are simple point processes, meaning node positions are distinct almost everywhere. Therefore, the convention will be that a node located at position $y\in\mathbb{R}^2$ will simply be referred to as node $y$. The $\ell$-2 norm will be denoted by $\|\cdot\|$. Random elements will generally be represented with an uppercase letter, a realization of which will be represented with a lowercase letter. For example, a realization of a point process $\Phi$ is $\phi$. For some set $\mathcal{W}$, the operator $[\cdot]_k$ produces $\left[\mathcal{W}\right]_k=\left\{A\subseteq\mathcal{W}\,\text{s.t.}\,|A|=k\right\}$, the set of subsets of $\mathcal{W}$ with cardinality $k$.
\begin{table}
\centering
\rowcolors{2}{gray!25}{white}
\begin{tabular}{@{}p{0.13\linewidth}p{0.6\linewidth}@{}} \toprule
\multicolumn{2}{c}{Common notation} \\ \cmidrule(r){1-2}
\rowcolor{white}
Notation & Description \\ \midrule
$\Phi$ & \begin{tabular}[x]{@{}l@{}}Poisson Point process in $\mathbb{R}^2$ composed of\\ two independent processes $\Phi\stackrel{\Delta}{=}\Phi_N\cup\Phi_I$\end{tabular} \\
$\lambda$ & \begin{tabular}[x]{@{}l@{}}Intensity of $\Phi_I$ and $\Phi_N$ \end{tabular}\\
$b_2(x,r)$ & \begin{tabular}[x]{@{}l@{}}Disk in $\mathbb{R}^2$ centered at $x$ with radius $r$\end{tabular} \\
$H_{ij}$ & \begin{tabular}[x]{@{}l@{}}Channel fading coefficient between a \\transmitter $i$ and receiver $j$\end{tabular}\\
$\mu^{\Phi_I}_{ji}$ & \begin{tabular}[x]{@{}l@{}}Probability of successful delivery of packet\\ from $j$ to $i$ in the presence of interferers $\Phi_I$\end{tabular}\\
$\theta$ & \begin{tabular}[x]{@{}l@{}}SIR threshold value; $\theta>1$\end{tabular}\\
$p$ & \begin{tabular}[x]{@{}l@{}}medium access probability common to all \\ nodes and interferers, including the base\\ station when broadcasting\end{tabular}\\
$\beta$ & \begin{tabular}[x]{@{}l@{}}Path loss exponent\end{tabular}\\
$\mathcal{O}$ & \begin{tabular}[x]{@{}l@{}}Base station situated at the origin; $\mathcal{O}=(0,0)$\end{tabular}\\
$X_i[t]$ & \begin{tabular}[x]{@{}l@{}}Inter-packet reception duration for the packet \\reception process of receiver $i$ at time $t$\end{tabular}\\

$A_{ji}(k)$ & \begin{tabular}[x]{@{}l@{}}AoI at receiver $i$ transmitted from $j$\end{tabular}\\
$\ell(x)$ & \begin{tabular}[x]{@{}l@{}}Path loss function\\ $\ell(x)=\|x\|^{-\beta},\ x\in\mathbb{R}^2$\end{tabular}\\
 \bottomrule
\end{tabular}
    \caption{}
    \label{tab:notation}
\end{table}

\begin{figure}
    \centering
    \includegraphics[width=0.5\textwidth]{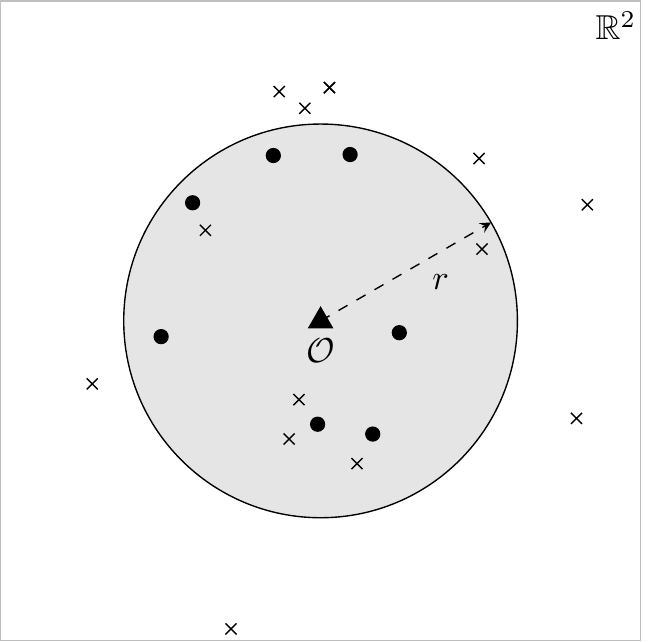}
    \caption{Example of a spatial realization $\phi_N$ of nodes (black circles), confined to a disk $b_2(0,r)$, and interferers (crosses) $\phi_I$, distributed across the Euclidean plane, with the base station (black triangle) in the center}
    \label{fig:ppp}
\end{figure}

\subsection{Network Model}\label{subsec:net_model}
Consider a base station, denoted by $\mathcal{O}$, in the Euclidean plane situated at the origin, with a finite set of nodes randomly distributed in a disk $b_2(0,r)$ of finite radius $r$. The nodes are distributed within a disk as a homogeneous Poisson Point Process with intensity $\lambda$, denoted by $\Phi_N$. Interferers are also distributed according to a homogeneous Poisson Point Process $\Phi_I$ that is distributed across $\mathbb{R}^2$ with intensity $\lambda$ (see~\Cref{fig:ppp}). We denote the combined point process of nodes and interferers by 
\begin{align}\Phi\stackrel{\Delta}{=}\Phi_N\cup\Phi_I\,.\end{align} 
This spatial model captures a scenario in which the base station may be one of many broadcast and collection nodes in a spatially-large wireless network, and where the base station is only interested in communicating with nodes within its vicinity. Each information update consists of a single, timestamped packet. When broadcasting, the base station attempts transmission of a packet to all nodes in the disk; the packet is successfully received at a receiver if the signal-to-interference ratio (SIR) exceeds a fixed threshold $\theta>1$. Similarly, during collection, the base station successfully receives a packet from a given transmitter when the SIR exceeds $\theta$.  All transmission attempts occur at the start of discrete time-slots, the packet duration and the slot length both normalized to $1$. Therefore, time $t$ is defined to be discrete, denoting the $t$\textsuperscript{th} slot. Medium access is granted to a transmitter -- including the base station when transmitting-- via an ALOHA-type random access scheme with a fixed common medium access probability (MAP) of $p$. That is, in any given time-slot the probability that a given transmitter attempts transmission is $p$, independent of all other time-slots and users in the network. In all subsequent sections we assume the packet delivery process is at steady state, having started at time $t=-\infty$.

The transmission power from every transmitter, including interferers, is fixed and normalized to 1. The wireless channel experiences Rayleigh fading and path loss attenuation. The fading loss random variable $H$ is i.i.d. exponentially distributed with mean 1. For a transmission from a transmitter $x$ to a receiver $y$, the path loss is defined to be $$\ell(x-y)\stackrel{\Delta}{=}\|x-y\|^{-\beta}\,.$$

The path loss exponent $\beta$ is generally chosen to be in the interval $(2,4)$. At time-slot $t$ the medium access indicator random variable $Z_x[t]$ is 1 if a transmitter $x$ attempts transmission and 0 otherwise. Given the realization of node and interferer locations $\phi$ and including medium access probability, transmission power, fading, and path loss, we may represent the signal power observed at receiver $y$ for a broadcast from the base station to be $$S^{\phi_I}_{\mathcal{O}y}[t]=Z_\mathcal{O}[t] H_{\mathcal{O}y}[t]\ell(y).$$ Similarly, the interference observed at $y$ is given by $$I^{\phi_I}_{\mathcal{O}y}[t]=\sum_{x\in\phi_I}Z_x[t] H_{xy}[t]\ell(x-y)\,.$$ Therefore, the SIR is given by the ratio of $S^\phi_{\mathcal{O}y}[t]$ and $I^\phi_{\mathcal{O}y}[t]$,
\begin{align}\label{eq:sir_b}
    SIR^{\phi_I}_{\mathcal{O}y}[t]=\frac{S^{\phi_I}_{\mathcal{O}y}[t]}{I^{\phi_I}_{\mathcal{O}y}[t]}=\frac{Z_\mathcal{O}[t]H_{\mathcal{O}y}[t]\ell(y)}{\sum_{x\in\phi_I}Z_x[t] H_{xy}[t]\ell(x-y)}\,.
\end{align}

For collection, the transmission signal from a transmitter $x$ in $\phi_N$ is subject to interference from both the field of interferers as well as other transmitters in $\phi_N$. Therefore, the SIR is
\begin{align*}
    SIR^{\phi}_{x\mathcal{O}}[t]=\frac{Z_x[t]H_{x\mathcal{O}}[t]\ell(x)}{\sum_{y\in\phi\setminus x}Z_y[t] H_{y\mathcal{O}}[t]\ell(y)}\,.
\end{align*}

In the following subsection, we formally define the AoI metric, which will then be used to define AoB and AoC.

\subsection{Age of Information}
AoI is denoted by $A[t]$. Let $G[t]$ be the time stamp of the most recent packet successfully received as of time $t$. The time evolution of AoI is then defined in~\Cref{eq:aoi}:
\begin{alignat}{1}\label{eq:aoi}
A[t+1]=\begin{cases}A[t]+1,\ &\text{if no reception} \\ \min\{t-G[t],\,A[t]\}+1,\ &\text{if reception}\end{cases}\,.
\end{alignat}

The AoI at a receiver $x$ and at time $t$ corresponding to information updates from some node $y$ is denoted by $A_{yx}[t]$.

\begin{figure}
    \centering
    \includegraphics[width=0.48\textwidth]{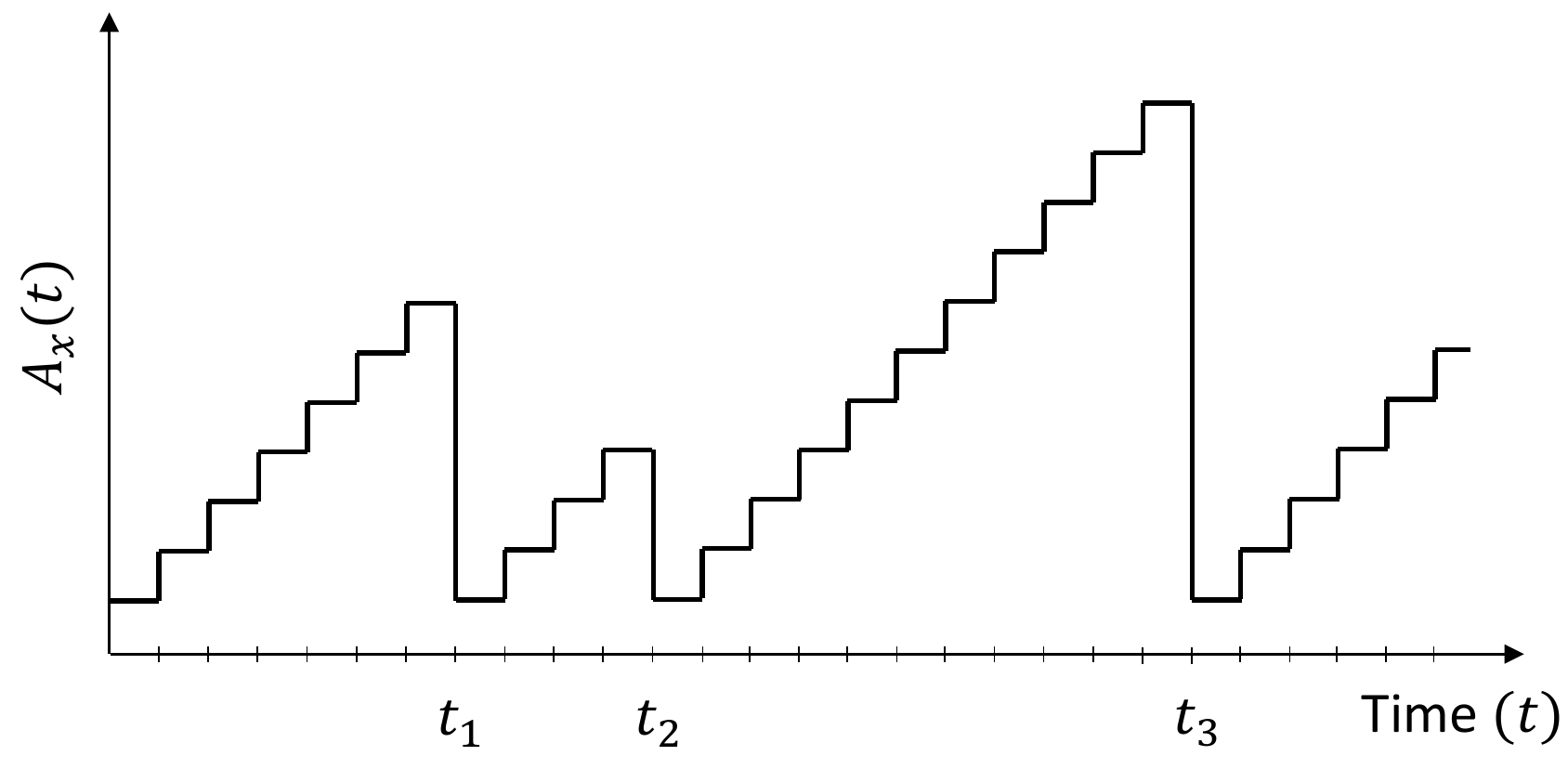}
    \caption{Age evolution over discrete time-slots}
    \label{fig:age_evolution}
\end{figure}

We assume any information source node can generate an update at-will, i.e. at each time $t$, an information update packet is instantaneously generated and transmitted with probability $p$. Therefore, the time stamp associated with a packet transmitted at time $t$ will always be $t$, and~\Cref{eq:aoi} becomes
\begin{alignat}{1}\label{eq:aoi_mod}
A[t+1]=\begin{cases}A[t]+1,&\text{if no reception} \\ 1,&\text{if reception}\end{cases}\,.
\end{alignat}
An example of the age evolution over time is provided in~\Cref{fig:age_evolution}. 

\subsection{AoB and AoC}
Having defined AoI and the network model, we now define Age of Broadcast and Age of Collection.
\subsubsection{Age of Broadcast}\label{subsubsec: b age}
The AoB $B_\mathcal{O}^{\phi_N}[t]$ with respect to some realization of the receiver locations $\phi_N$ and base station $\mathcal{O}$ at time-slot $t$ is defined as 
\begin{align}
    B_\mathcal{O}^{\phi}[t]\stackrel{\Delta}{=}\max_{i\in\phi_N}A^{\phi_I}_{\mathcal{O} i}[t]\,.
\end{align}
Given the base station begins broadcasting at $t-B_\mathcal{O}^{\phi}[t]$, the time until all receivers get an update cannot be less than $B_\mathcal{O}^{\phi}[t]$. Moreover, at least one base station has an update that is no greater than $t-B^{\phi}_{\mathcal{O}}[t]$



\subsubsection{Age of Collection}\label{subsubsec: c_age}
The AoC $C_{\mathcal{O}}^{\phi}[t]$ with respect to the base station $\mathcal{O}$ and the realization of node and interferer locations $\phi$ at time $t$ is defined as
\begin{align}
    C_{\mathcal{O}}^{\phi}[t]\stackrel{\Delta}{=}\max_{j\in\phi_N}A^{\phi}_{j \mathcal{O}}[t]\,.
\end{align}

For a given time $t$, the base station will have received at least one update from all but one transmitter since $t-C^{\phi}_{\mathcal{O}}[t]$.


In both broadcast and collection settings, we adopt the convention that if $\phi_N=\emptyset$, then AoB and AoC is 0 for all time. Having defined AoB and AoC, we establish preliminary results that are used in subsequent sections to analyze AoB and AoC.


%% file: texfiles/preliminaries.tex
\section{Preliminaries}\label{sec:preliminaries}
When broadcasting, the probability the base station successfully delivers an update to an arbitrary receiver $y$  given the locations of the interferer positions $\phi_I$ is determined by the medium access probability and the channel characteristics. Since transmission attempts from transmitters and interferers alike are i.i.d. and Bernoulli in each time slot, the probability of successful delivery to $y$ is time-invariant and the time index $t$ can be dropped. Given the spatial realization of interferers $\phi_I$, the success probability is given by,
\begin{align}
    \mu^{\phi_I}_{\mathcal{O}y}=\mathbb{P}\left(SIR_{\mathcal{O}y}>\theta\right)\,.
\end{align}

Averaging over the channel fading and the random access, and given the spatial realization $\phi_I$, the conditional reception success probability at a receiver $y$ is given by
\begin{align}
    \mu_{\mathcal{O}y}^{\phi_I}&=\mathbb{P}\left(SIR^{\phi_I}_{\mathcal{O}y}>\theta\,\big|\,\Phi_I=\phi_I\right)\\&=p\prod_{x\in\phi_I}\left(1-\frac{p}{1+\theta\frac{\ell(y)}{\ell(x-y)}}\right)\,.\label{eq:lemma_succ}
\end{align}
The derivation of~\Cref{eq:lemma_succ} is omitted due to space constraints but is similar to the analysis in~\cite{yang_optimizing_2020} Lemma 1.
Through an identical line of reasoning, the conditional success probability during collection with respect to transmitter $x\in\phi_N$ is given by
\begin{align}\label{eq:coll_cond_succ}
    \mu^\phi_{x\mathcal{O}}=p\prod_{y\in\phi\setminus x}\left(1-\frac{p}{1+\theta\frac{\ell(x)}{\ell(y)}}\right)\,.
\end{align}

where the sources of interference are both $\phi_I$ and $\phi_N\setminus x$; thus success probability is conditioned on $\phi$ instead of $\phi_I$.

Note that $\mu^{\phi_I}_{\mathcal{O}y}$ and $\mu^{\phi}_{x\mathcal{O}}$ are dependent on the realization $\phi$. Thus, when not given $\phi$, the reception success probability is a random variable. 

Next, we de-condition~\Cref{eq:lemma_succ} and~\Cref{eq:coll_cond_succ} on $\Phi_I$ by taking the average over all realizations of the interferer locations. The packet reception success probability from the base station $\mathcal{O}$ to a receiver $y$ is then given by
\begin{align}\label{eq:av_succ}
    \mu(\|y\|)= p\exp\left(-p\lambda\pi C \|y\|^2\right)\,,
\end{align}
where
\begin{align}
 C\triangleq\Gamma(1+\delta)\Gamma(1-\delta)\theta^\delta\label{eq:c_const}
\end{align}
and the gamma function  $\Gamma(\cdot)$ is defined as ${\Gamma(x)\stackrel{\Delta}{=}\int_0^\infty t^{x-1}e^{-t}\,dt}$, and $\delta=\frac{2}{\beta}$.

The proof is omitted for brevity, but is similar to the analysis in~\cite{haenggi_interference_2008} (see Section 3.2.3).
Note that $\mu$ no longer depends explicitly on node or interferer geometry and is instead only a function of the distance between $\mathcal{O}$ and $y$. Thus, in the instance-independent analysis we will express the success probability purely as a function of distance between the base station and the node. 

Conditioned on $\Phi_N$, we find the success probability in the collection case de-conditioning on $\Phi_I$ to be
\begin{align}
    \mu^{\phi_N}_{y\mathcal{O}}&=\mathbb{P}\left(H_{y\mathcal{O}}\geq \frac{\theta\left(I^{\phi_I}_{y\mathcal{O}}+I^{\phi_N\setminus\{y\}}_{y\mathcal{O}}\right)}{\ell(y)}\right)\\
    &=\mathbb{E}\left[\exp\left(-\frac{\theta I^{\phi_I}_{y\mathcal{O}}}{\ell(y)}\right)\right]\cdot \mathbb{E}\left[\exp\left(-\frac{\theta I^{\phi_N\setminus\{y\}}_{y\mathcal{O}}}{\ell(y)}\right)\right]\\
    &=p\exp\left(-p\lambda\pi C \|y\|^2\right)\cdot\prod_{j\in\phi_N\setminus\{y\}}\left(1-\frac{p}{1+\theta\frac{\ell(y)}{\ell(j)}}\right)\\
    &=\mu(\|y\|)\cdot\prod_{j\in\phi_N\setminus\{y\}}\left(1-\frac{p}{1+\theta\frac{\ell(y)}{\ell(j)}}\right)\,,
\end{align}
where $I^{\phi_N\setminus\{y\}}_{y\mathcal{O}}$ denotes the interference induced by the transmission of the nodes in $\phi_N\setminus\{y\}$.
Having established packet reception probabilities results in both instance-dependent and instance-independent cases, we leverage this insight in deriving AoB and AoC in the following sections, starting with broadcast.

%% file: texfiles/broadcast.tex

\section{Broadcast}\label{sec:broadcast}
In this section we characterize the Expected AoB (EAoB), the expectation taken with respect to the ALOHA network traffic. In the instance-dependent case, we analyze EAoB given perfect knowledge of node and interferer locations. In the instance-independent case, node and interferer locations are unknown, so we find EAoB in expectation over the node and interferer point processes.

\subsection{Instance-dependent (BD)}\label{subsec:bd}
We assume the locations of all interferers and receivers are known. Interferers are located according to a realization of $\Phi_I$, denoted $\phi_I$. Receivers are also distributed according to a realization of $\Phi_N$ and is denoted $\phi_N$. Recall that the network is at steady state, having started at $t=-\infty$. Therefore, the AoB process is stationary and EAoB, defined as $\mathbb{E}\left[B^{\phi}_{\mathcal{O}}[t]\right]$, is the same for all finite $t$ and the dependence on time can be dropped to give $\mathbb{E}\left[B^{\phi}_\mathcal{O}\right]$.

\begin{figure}
\centering
     \begin{subfigure}[b]{0.5\textwidth}
         \centering
         \includegraphics[width=\textwidth]{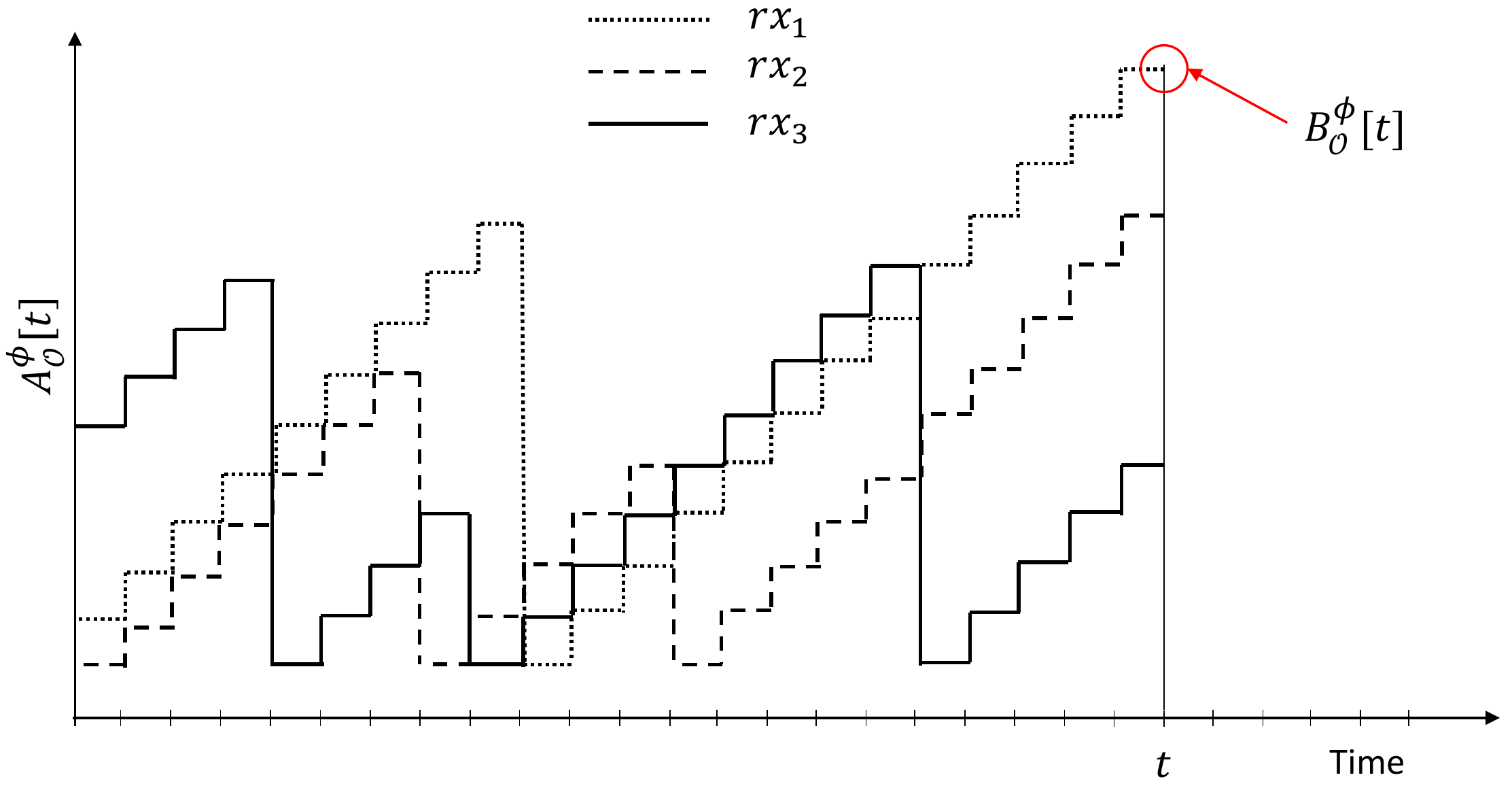}
         \caption{}
         \label{fig:3_pkt_rcp}
     \end{subfigure}
     \begin{subfigure}[b]{0.5\textwidth}
         \centering
         \includegraphics[width=\textwidth]{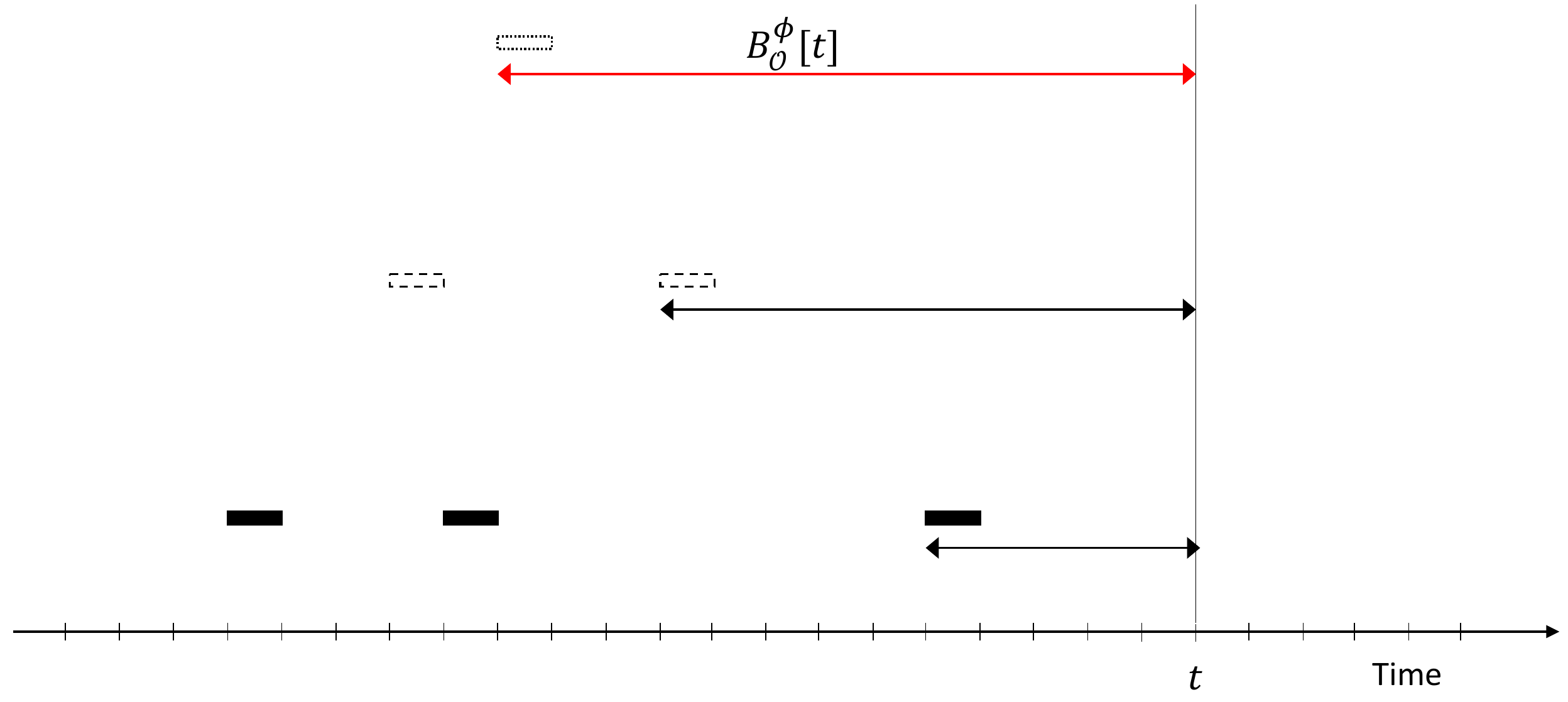}
         \caption{}
         \label{fig:forward_process}
     \end{subfigure}
     \begin{subfigure}[b]{0.5\textwidth}
         \centering
         \includegraphics[width=\textwidth]{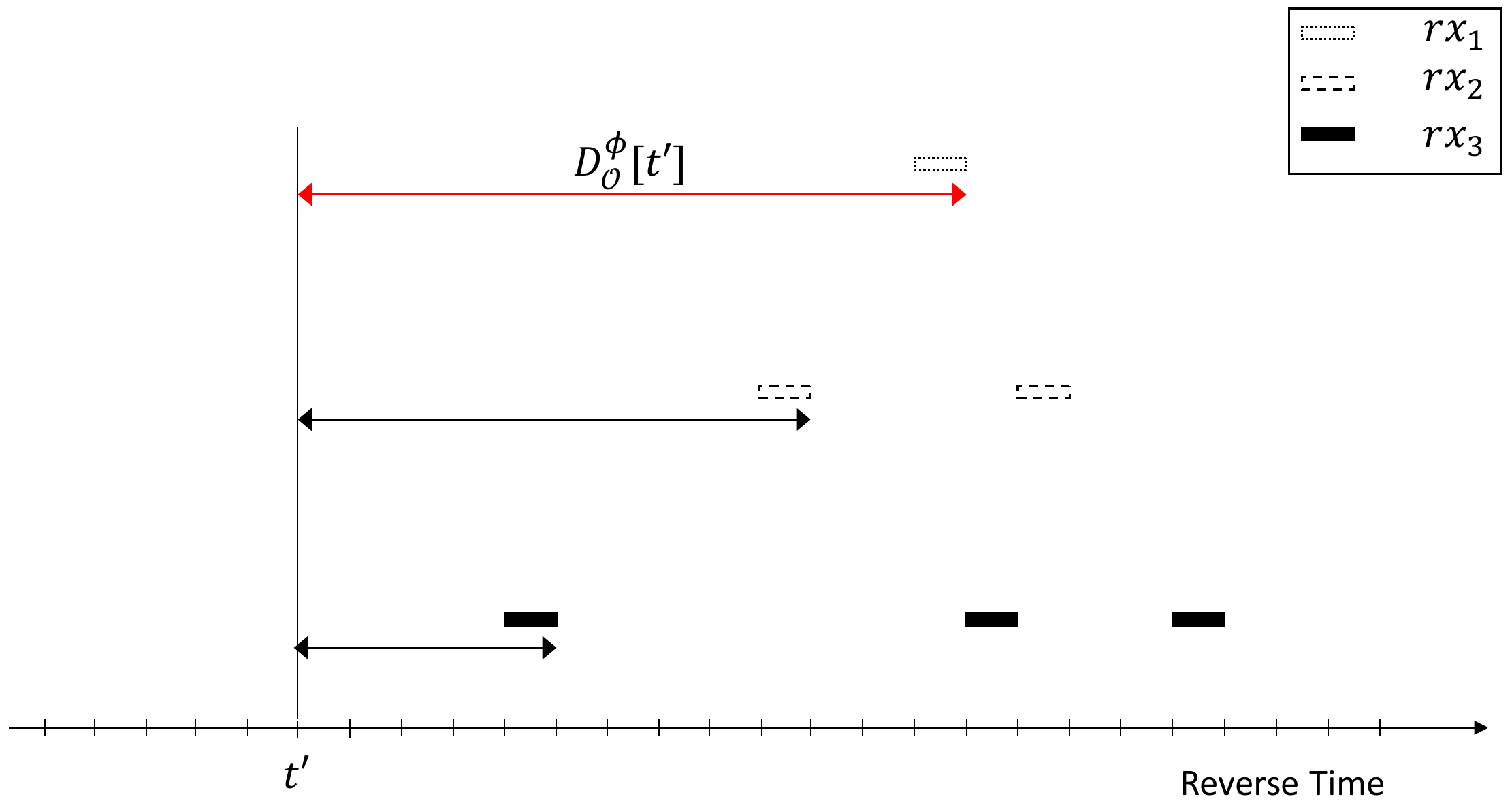}
         \caption{}
         \label{fig:reverse_process}
     \end{subfigure}
        \caption{An illustrative example supporting~\Cref{cl:br_delay}, where (a) is the age evolution of three receivers that form the set $\phi_N$, (b) describes the packet arrival process up to time $t$, and (c) is the reverse of the process in (b).}
        \label{fig:three graphs}
\end{figure}

To determine the EAoB, it is helpful to connect average broadcast age to the average broadcast delay, a related-yet-distinct metric~\cite{xie2013network}. For each receiver  $y\in\phi_N$, define $X^{\phi_I}_{\mathcal{O}y}[t]$ to be the time elapsed until successful reception of the next packet at receiver $i$ since time $t$. Broadcast delay $D^{\phi
}_\mathcal{O}[t]$ is then the time elapsed from time $t$ until the time at which all receivers in $\phi_N$ have received the next packet. That is,
\begin{align}\label{eq:bdelay}
    D^{\phi}_\mathcal{O}[t]\triangleq\max_{y\in\phi_N}X^{\phi_I}_{\mathcal{O}y}[t]
\end{align}

Since the packet reception process is stationary -- by virtue of the ALOHA random access and i.i.d. fading --  the average broadcast delay $\mathbb{E}\left[D^{\phi}_\mathcal{O}[t]\right]$ is the same for all finite $t$ and the dependence on $t$ can be dropped to give $\mathbb{E}\left[D^{\phi}_{\mathcal{O}}\right]$. We now reason that the average broadcast delay is equivalent to EAoB. This is evident by observation of the reverse packet reception process (see~\Cref{fig:three graphs}). 
Consider the instantaneous AoB at time $t$, which is the maximum AoI in $\phi_N$, and is exactly the time that elapsed between the current time point and the first time point in the past since which all receivers in the receiver set have successfully received at least one information update (see~\Cref{fig:forward_process}). Now consider the broadcast delay, given by~\Cref{eq:bdelay} and shown in~\Cref{fig:reverse_process}. Since the packet reception process's evolution in the forward direction is identical in distribution to that of its reverse process looking back in time, we may conclude the average broadcast delay must be equal to the EAoB. The formal claim is outlined below.
    
\begin{claim}\label{cl:br_delay}
 
\begin{align}\mathbb{E}\left[B^{\phi}_\mathcal{O}\right] =\mathbb{E}\left[D^{\phi}_\mathcal{O}\right]\,.\end{align}
\end{claim}
\begin{proof} See Appendix~\Cref{app:br_delay_pf}. Refer to~\Cref{fig:three graphs} for an intuitive illustration of the proof.
\end{proof}

This equivalence between the average broadcast delay and average AoB facilitates an analytical derivation of expected AoB. We begin by first defining the joint distribution of packet reception at each time $t$. Assuming knowledge of the locations of all nodes in $\phi_I$ and ${\phi_N}$, at time slot $t$, the joint distribution of packet reception for all the receivers, i.e. the joint distribution of $\left\{\mathbbm{1}^{\phi_I}_{\mathcal{O}i}[t]\right\}_{i\in\phi_N}$, the packet reception indicator random variables, can be obtained. Using this joint distribution, it is possible to determine the average broadcast age explicitly.


We partition the receivers into $\Xi[t]=\{i\in\phi_N\,|\,\mathbbm{1}_{\mathcal{O}i}[t]=1\}$ and $\Psi[t]=\{j\in\phi_N\,|\,\mathbbm{1}_{\mathcal{O}j}[t]=0\}$, the set of receivers that successfully received a packet at time $t$ and the set that did not receive a packet, respectively. Defining the probability of the set of receivers $\mathcal{R}$ all successfully receiving a packet at time slot $t$ as the following,
\begin{align}
    \mu_{\mathcal{O}\mathcal{R}}^{\phi_I}[t]&=\mathbbm{P}\left(\cap_{i\in\mathcal{R}}\left\{SIR^{\phi_I}_{\mathcal{O}i}[t]>\theta\right\}\right)\\
    &=\mathbbm{P}\left(\cap_{i\in\mathcal{R}}\left\{\mathbbm{1}^{\phi_I}_{\mathcal{O}i}[t]=1\right\}\right)\, ,
\end{align}
and conversely $w^{\phi_I}_{\mathcal{O}\mathcal{R}}[t]$ as the probability of the set of receivers $\mathcal{R}$ NOT receiving a packet at time slot $t$ as
\begin{align}
    w^{\phi_I}_{\mathcal{O}\mathcal{R}}[t]&=\mathbbm{P}\left(\cap_{i\in\mathcal{R}}\left\{SIR^{\phi_I}_{\mathcal{O}i}[t]\leq\theta\right\}\right)\\
    &=\mathbbm{P}\left(\cap_{i\in\mathcal{R}}\left\{\mathbbm{1}^{\phi_I}_{\mathcal{O}i}[t]=0\right\}\right)\, .
\end{align}
 We next determine the probability rule of $\left\{\mathbbm{1}^{\phi_I}_{\mathcal{O}i}[t]\right\}_{i\in\phi_N}$, which is equivalent to finding the probability rule for $\Xi[t]$ without loss of generality. By the inclusion-exclusion property, the probability rule for $\Xi[t]$ can be established. 
 
 The inclusion-exclusion principle formula represents the probability of the union of a set of events ${A=\left\{A_1,\,A_2,\,\hdots,\,A_n\right\}}$ by the probabilities of intersections of its subsets: 
 \begin{align}
     \mathbb{P}\left(\cup_{i= 1}^n A_i\right)=\sum_{j=1}^n (-1)^{n-1}\sum_{L\in[A]_k}\mathbb{P}\left(\cap_{i\in L}A_i\right)\,,
 \end{align}
where $[A]_k$ denotes the set of subsets of $A$ with cardinality $k$.
  Therefore, the probability rule for $\Xi[t]$ is
 \begin{align}
       p\left(\Xi[t]\right)&=\mathbb{P}\left(\left\{\cap_{i\in \Xi}\mathbbm{1}^{\phi_I}_{\mathcal{O}i}[t]=1\right\}\bigcap \left\{\cap_{j\in \Psi}\mathbbm{1}^{\phi_I}_{\mathcal{O}j}[t]=0\right\}\right)\\
       &=\mathbb{P}\left(\cap_{i\in\Xi}\left\{\mathbbm{1}^{\phi_I}_{\mathcal{O}i}[t]=1\right\}\right)\\
       &\quad-\mathbb{P}\left(\cup_{j\in\Psi}\left\{\mathbbm{1}^{\phi_I}_{\mathcal{O}j}[t]=1\right\}\cap\left\{\cap_{i\in\Xi}\mathbbm{1}^{\phi_I}_{\mathcal{O}i}[t]=1\right\}\right)\\
&=\mu^{\phi_I}_{\mathcal{O}\Xi}[t]+\sum_{k=1}^{|\Psi|}(-1)^k\sum_{L\in [\Psi]_k}\mu_{\mathcal{O}\{\Xi\cup L\}}^{\phi_I}[t]\,.
 \end{align}

Without loss of generality, we may define the probability rule starting at time $t=0$ for broadcast delay $D^{\phi}_{\mathcal{O}}[0]$. For the broadcast delay to be $\tau$, at least one receiver must receive a packet at time $t=\tau-1$. Moreover, of all the receivers that received a packet at time $\tau-1$, at least one must have received no packets for times $t\in\{0,\,\hdots,\,\tau-2\}$, i.e. given non-empty $\Xi[\tau-1]$, there exists some non-empty subset $J\subseteq{ \Xi[\tau-1]}$ such that $J\subseteq\cap_{i=0}^{\tau-2}\Psi[i]$. Therefore, $p^\phi_D(\tau)$, the probability that the broadcast delay is equal to $\tau$, is given by
\begin{align}
   p_{D}^\phi(\tau)&=\sum_{\Xi\in 2^{\Phi_N}\setminus\{\emptyset\}}p(\Xi[\tau-1])\\&\cdot\underbrace{\mathbb{P}(\cup_{J\in 2^{\Xi[\tau-1]}\setminus\{\emptyset\}}\{J\subseteq\cap_{i=0}^{\tau-2}\Psi[i]\})}_{(*)}\\
    &=\sum_{\Xi\in 2^{\Phi_N}\setminus\{\emptyset\}}p(\Xi[\tau-1])\\
    &\quad\cdot\underbrace{\left(\sum_{n=1}^{|\Xi|}(-1)^{n+1}\sum_{J\in[\Xi]_n}\prod_{u=0}^{\tau-2} w^{\phi_I}_{\mathcal{O}J}[u]\right)}_{\text{(\textdagger)}}\,,
\end{align}

where $2^{\{\cdot\}}$ denotes the power set of some set $\{\cdot\}$, and (\textdagger) is the inclusion-exclusion formula applied to $(*)$.  Finally, we use~\Cref{cl:br_delay} to find EAoB given by
\begin{align}
    \mathbb{E}\left[B^{\phi}_\mathcal{O}[t]\right] &=\mathbb{E}\left[D^{\phi}_\mathcal{O}\right]\\
    &=\mathbb{E}\left[D^{\phi}_\mathcal{O}[0]\right]=\sum_{k=1}^\infty k\cdot p_D^\phi(k)\,.
\end{align}

While this representation of average broadcast age is complete, a more intuitive characterization can be developed in the form of bounds on the average broadcast delay. We begin with an empirical observation. Based on simulation results, we observe that the average broadcast age is bounded above by the average broadcast age of an alternate packet reception process in which all packet reception indicators were independent random variables, albeit preserving the same distribution (see~\Cref{fig:associated}). Recall that $X^{\phi_I}_{\mathcal{O}i}[t]$ denotes the time elapsed since time $t$ until the next packet reception. The conjecture based on this observation is formalized as follows.

\begin{conjecture} The EAoB is bounded above by
\begin{align*}
    \mathbb{E}\left[B^{\phi}_{\mathcal{O}}\right]\leq\mathbb{E}\left[\Tilde{D}^{\phi}_{\mathcal{O}}\right]=\mathbb{E}\left[\max_{i\in\phi_N}\Tilde{X}^{\phi_I}_{\mathcal{O}i}[0]\right]\,,
\end{align*}
 where $\Tilde{X}^{\phi_I}_{\mathcal{O}i}[0]\stackrel{d}{=}X^{\phi_I}_{\mathcal{O}i}[0]$ and $\left\{\Tilde{X}^{\phi_I}_{\mathcal{O}i}[0]\right\}_{i\in\phi_N}$ are independent.
\end{conjecture}


\begin{figure}
    \centering
    \includegraphics[width=1.1\linewidth]{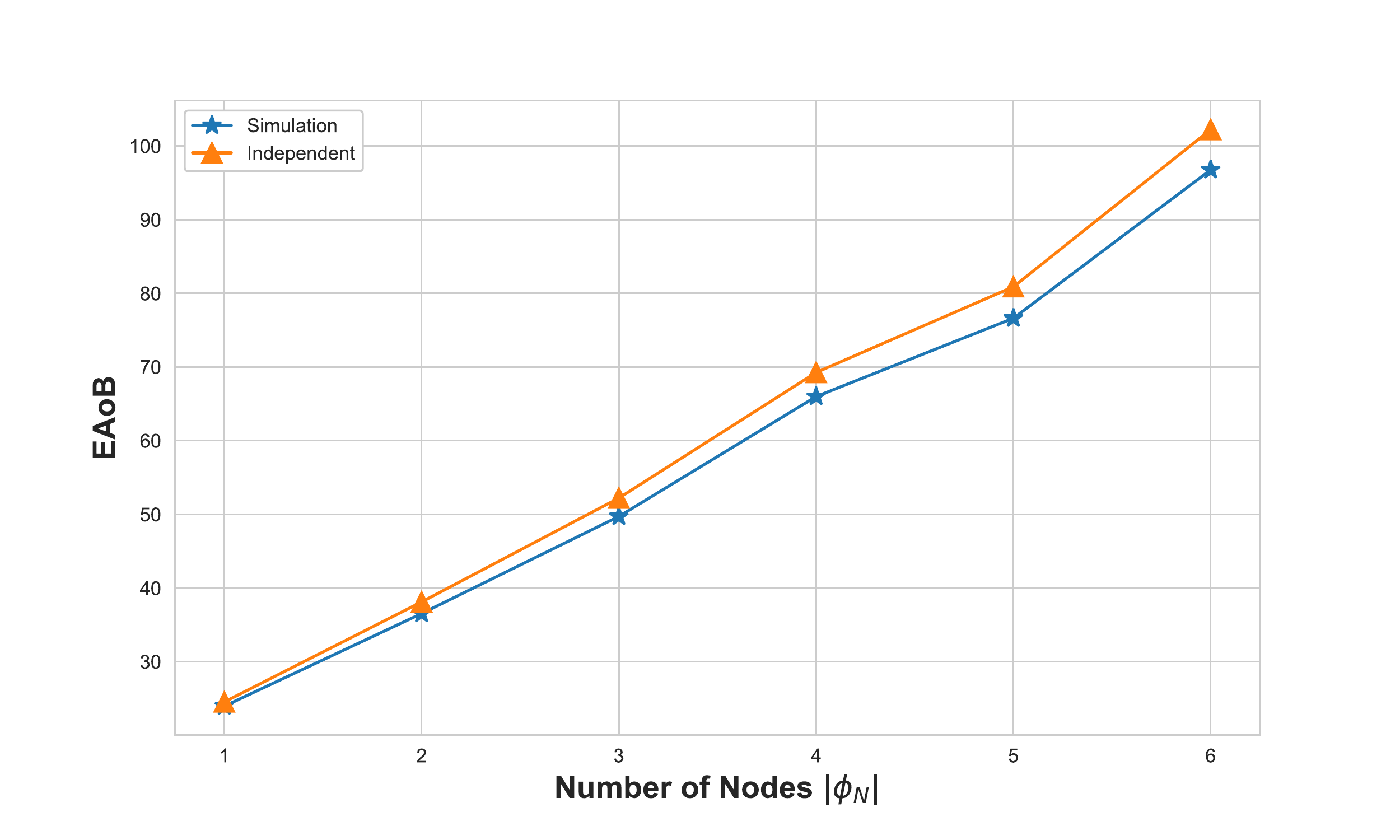}
    \caption{Simulation of Broadcast age compared against the maximum of independent random variables. The actual simulation's EAoB is consistently less than the independent random variable counterpart.}
    \label{fig:associated}
\end{figure}

~\Cref{fig:associated} Illustrates the conjecture stated, plotting the simulated broadcast delay and the broadcast delay expected if all packet reception processes were independent, maintaining the same packet reception distributions.

Having established an explicit formulation of EAoB as well as a conjectured upper bound, we may turn to the more general, instance-independent regime. Upper bounds are pursued in the instance-independent scenario, as outlined in the following subsection.

\subsection{Instance-independent (BI)}\label{subsec:bi}

In~\Cref{subsec:bd} we considered a particular instantiation of the nodes. Here, we take the expectation with respect to node and interferer positions. An upper bound can be found with a differential equation approach. Consider the squared ordered distances of the $i$ nearest receivers to $\mathcal{O}$ denoted $ R_{1}^2\leq R_2^2\leq\,\hdots\,\leq R_{i}^2$.
In two-dimensional Poisson Point Processes such as $\Phi_N$, the squared ordered distances have the same distribution as that of the arrival times in a one-dimensional Poisson Process $\Phi_N'\subset\mathbb{R}^+$ of intensity $\lambda'=\lambda\pi$~\cite{mathar1995distribution} 

Focusing on this one-dimensional point process, consider a small interval in $\mathbb{R}^+$ given by $(x,x+\Delta]$ for very small $\Delta$. A receiver $y\in\Phi_N$ exists in the point process $\Phi_N'$ in the interval $(x,x+\Delta]\subset \mathbb{R}^+$ with probability $\lambda'\Delta$. We define the EAoB over the set of receivers in $\Phi_N$ that map to $(0,u]\in\mathbb{R}^+$ to be ${\overline{B}(u)}$. If a receiver does not exist in $(x,x+\Delta]\subset \mathbb{R}^+$, then the EAoB $\overline{B}(x+\Delta)$ would be the same as $\overline{B}(x)$. If a receiver $y$ does exist in the interval, either $y$ receives a packet after all the other receivers in $(0,x]$ with probability $\varsigma$ or it does not with probability $1-\varsigma$.
By setting the probability of $y$ getting a packet after the rest of the receivers to $\varsigma=1$, we upper bound the time to broadcast to all receivers that are in $(0,x+\Delta]$, as shown in the following:
\begin{align}
    \overline{B}(x+\Delta)&=(1-\lambda'\Delta)\overline{B}(x)\\  
    &+(\lambda'\Delta)\left( \overline{B}(x)+\varsigma\frac{1}{\mu\left(\sqrt{x+\Delta}\right)}\right)\\
    &=\overline{B}(x)+(\lambda'\Delta)\varsigma\frac{1}{\mu\left(\sqrt{x+\Delta}\right)}\\
    &\leq \overline{B}(x)+(\lambda'\Delta)\frac{1}{\mu\left(\sqrt{x+\Delta}\right)}\,,\label{eq:diff_ub}
\end{align}
where $\mu(\cdot)$ is given by~\Cref{eq:av_succ}. As mentioned before, we upper bound in~\Cref{eq:diff_ub} by setting $\varsigma$ to 1. The average packet reception delay  $\frac{1}{\mu\left(\sqrt{x+\Delta}\right)}$ is readily found by taking the reciprocal of $\mu(\sqrt{x+\Delta})$ as given by~\Cref{eq:av_succ} since the packet reception process is i.i.d Bernoulli. By bringing $\overline{B}(x)$ over to the left-hand side of the equation, dividing both sides by $\Delta$ and taking the limit as $\Delta\to 0$, we arrive at the following:
\begin{align}
        &\lim_{\Delta\to 0}\frac{\overline{B}(x+\Delta)-\overline{B}(x)}{\Delta}\leq \lim_{\Delta\to 0}\frac{\lambda'}{\mu\left(\sqrt{x+\Delta}\right)}\\
        &\frac{d\overline{B}}{dx}\leq \frac{\lambda'}{\mu\left(\sqrt{x}\right)}=\frac{\lambda\pi}{p}\exp\left(p\lambda\pi C x\right)\\
        &\frac{d\overline{B}}{dr}=\frac{d\overline{B}}{dx}\frac{dx}{dr}\leq  \frac{\lambda'}{\mu\left(r\right)} \cdot 2r=\frac{2\lambda\pi r}{p}\exp\left(p\lambda\pi C r^2\right) \label{eq:diff_intermediate}
    \end{align}

Solving the differential equation in~\Cref{eq:diff_intermediate} with the initial condition $B(0)=0$, we obtain
\begin{align}
    B(r)\leq \frac{1}{p^2C}\left(e^{p\lambda\pi C r^2}-1\right)
\end{align}


\begin{figure}
    \centering
\includegraphics[width=1.1\linewidth]{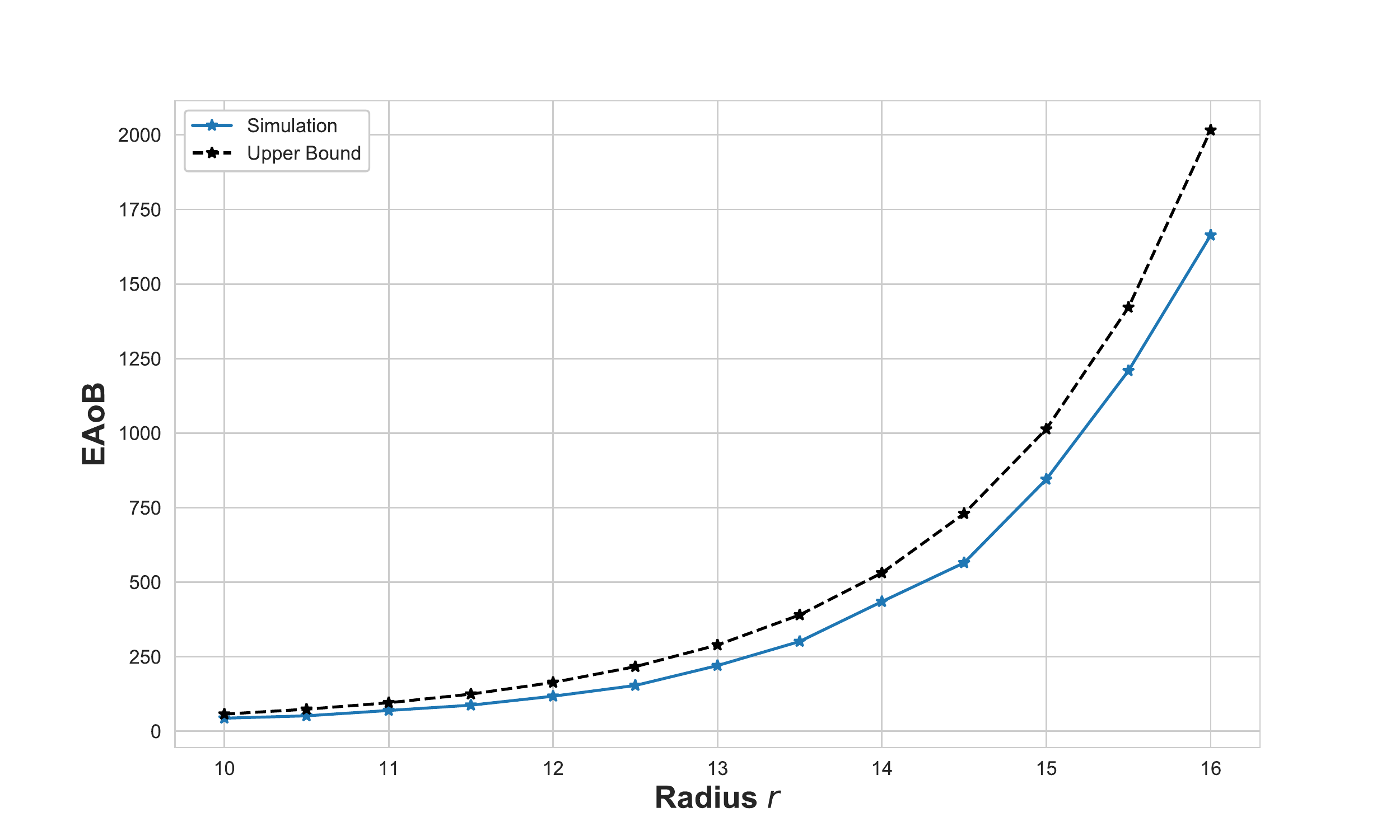}
    \caption{Instance-Independent AoB scaling}
    \label{fig:BI_scaling}
\end{figure}

~\Cref{fig:BI_scaling} compares the upper bound through the differential analysis above against simulation.

In the next section we characterize the collection problem and find bounds on the performance.

%% file: texfiles/collection.tex
\section{Collection}\label{sec:collection}

In the collection problem, the base station acts as a receiver situated at the origin, with a set of transmitters $\phi_N$ positioned within $b_2(0,r)$ sending updates to the base station. We characterize the expected age of collection (EAoC), defined as $\mathbb{E}\left[C^{\phi}_\mathcal{O}[t]\right]$. As in the broadcast section, the collection age is characterized in both instance-dependent and instance-independent regimes.

\subsection{Instance-dependent (CD)}

We begin the instance-dependent analysis by connecting AoC with a metric we denote as collection delay. For each transmitter $i\in\phi_N$, define $Y^{\phi}_{i\mathcal{O}}[t]$ to be the time elapsed between the current time $t$ and the successful reception of the next packet transmitted by $i$ to the base station after time $t$. The collection delay $\mathcal{K}^{\phi}_{\mathcal{O}}[t]$ is defined as
\begin{align*}
    \mathcal{K}^{\phi}_{\mathcal{O}}[t]\stackrel{\Delta}{=}\max_{i\in\phi_N}Y^{\phi}_{i\mathcal{O}}[t]
\end{align*}

Since the packet reception process is i.i.d. over time and the AoC process is stationary since the network began at ${t=-\infty}$, the time index can be dropped. By an identical line of reasoning as that in~\Cref{cl:br_delay}, we conclude that 
\begin{align}\label{eq:coll_delay}
    \mathbb{E}\left[C^\phi_{\mathcal{O}}\right]=\mathbb{E}\left[ \mathcal{K}^{\phi}_{\mathcal{O}}\right]\,.
\end{align}


Since $\theta >1$, the event of a packet reception at time $t$ at the base station from transmitter $x$ is disjoint from the event of a packet reception in the same time slot from transmitter $j\neq x$. That is, due to the threshold setting being larger than $1$, only a single packet can be received at the receiver in a single time slot. When packet reception events from different transmitters are disjoint and the packet reception process is time invariant, we observe that the update collection process resembles a coupon collection process.


In the classical Coupon Collector Problem, there are $n$ distinct coupons that are to be collected. Coupons are drawn randomly at each time step. At any time, the probability of drawing any one of $n$ coupons is uniformly $\frac{1}{n}$, independent of all other time steps, and so the resulting average time it takes to draw all $n$ distinct coupons at least once is $nH_n$, where $H_n$ denotes the $n$\textsuperscript{th} harmonic number $H_n=\sum_{k=1}^n\frac{1}{k}$.

The variant of the CCP in the collection scenario is one in which $|\phi_N|$ distinct coupons need to be drawn but do not have a uniform probability of being drawn by the base station~\cite{flajolet1992birthday}. Additionally, there is the possibility of drawing an unwanted NULL coupon -- the event where no packet is successfully received -- which occurs if either no transmitter attempts a transmission or all attempted transmissions failed to exceed $\theta$.

The expression for the average collection can be found, expressed in~\Cref{cl:cct}.

\begin{claim}\label{cl:cct}
Given knowledge of the node and interferer locations $\phi$ the EAoC is given by
\begin{align}\mathbb{E}\left[C_\mathcal{O}^\phi\right]&=\mathbb{E}\left[\max_{i\in\phi_N} Y^\phi_{i\mathcal{O}}[0]\right]\\&=\sum_{i=1}^n(-1)^{i+1}\sum_{A\in [\phi_N]_k}\frac{1}{\sum_{u\in A}\mu^\phi_{u\mathcal{O}}}\, .\end{align} 
\end{claim}
\begin{proof} 
Since we have established the equivalence of EAoC and expected collection delay in~\Cref{eq:coll_delay}, we focus on the expected collection delay. We invoke the maximum-minimums identity to represent $\max_{i\in\phi_N}Y^{\phi}_{i\mathcal{O}}$ as a sum of the minima of the non-empty subsets of $\phi_N$. The maximum-minimums identity states that for a finite set of numbers $A$ with cardinality $n$,
\begin{align}\label{eq:max_min_id}
    \max A=\sum_{k=1}^n (-1)^{k+1}\sum_{L\in [A]_k}\min {L}\,.
\end{align}
Applying the identity to $\max_{i\in\phi_N}Y^{\phi}_{i\mathcal{O}}$,
\begin{align}
    \mathbb{E}\left[\max_{i\in\phi_N}Y^{\phi}_{i\mathcal{O}}\right]=\sum_{k=1}^{|\phi_N|}(-1)^{k+1}\sum_{A\in[\phi_N]_k}\mathbb{E}\left[\min_{j\in A}Y^{\phi}_{j\mathcal{O}}\right]\, .
\end{align}
Due to the disjointedness of packet reception events between any transmitters in a time slot $t$, the random variable $\min_{j\in A}Y^{\phi}_{j\mathcal{O}}$ is a geometric random variable with parameter $\sum_{j\in A}\mu_{j\mathcal{O}}^\phi$. Thus,
\begin{align}
 \mathbb{E}[\max_{i\in\phi_N}Y^{\phi}_{i\mathcal{O}}]=\sum_{k=1}^{|\phi_N|}(-1)^{k+1}\sum_{A\in[\phi_N]_k}\frac{1}{\sum_{j\in A}\mu^\phi_{j\mathcal{O}}}\,,
\end{align}
 and the proof is complete.
\end{proof}

We proceed to find bounds on the EAoC in the instance-independent case in the following subsection.

\subsection{Instance-independent (CI)}

In this scenario, the locations of the transmitters are no longer assumed to be known, distributed according to the Poisson Point Process $\Phi_N$. We begin with an upper bound on EAoC. Conditioning on the size of $\Phi_N$ to be $n$, the nodes are distributed i.i.d. uniform in the disk $b_2(\mathcal{O},r)$. Based on this conditioning, and assuming no nodes are present within a small distance $\epsilon$ of the base station, an upper bound for EAoC is outlined in the following claim:

\begin{claim}\label{cl:ci_ub}
Conditioned on the number of transmitters $|\Phi_N|=n$ in the disk $b_2(\mathcal{O},r)$ the EAoC is bounded above as given by
\begin{align}
\mathbb{E}\left[C^\phi_\mathcal{O}\,\big|\, |\Phi_N|=n\right]\leq \left(\overline{\mu}\right)^{-1}H_n\, ,
\end{align}
where
\begin{align}
    \overline{\mu}=\left(1-\frac{p}{1+\theta 
    \left(\frac{\epsilon}{r}\right)^{\beta}}\right)^{n-1}\cdot \mu(r)
\end{align}

\end{claim}

\begin{proof}
The most disadvantaged transmitter in terms of successful delivery probability is one situated at a distance $r$ from the base station, while the remaining $n-1$ transmitters are close enough to the base station to observe no path loss, i.e. at distance $\epsilon$. The success probability of this disadvantaged transmitter is
\begin{align}
    \left(1-\frac{p}{1+\theta \left(\frac{\epsilon}{r}\right)^{\beta}}\right)^{n-1}\cdot\mu(r)\,,
\end{align}
since each of the other transmitters contribute equally to the interference observed at the base station.

If all transmitters have this same pessimistic delivery probability, the EAoC would be that of a classical CCP, the resulting EAoC given in~\Cref{cl:ci_ub}. 
\end{proof}


We now present numerical simulations that highlight the interplay between broadcast and collection.

%% file: texfiles/numerical.tex
\section{Numerical Results}\label{sec:numerical}
\begin{table}
\centering
\rowcolors{2}{gray!25}{white}
\begin{tabular}{@{}p{0.13\linewidth}p{0.6\linewidth}@{}} \toprule
\multicolumn{2}{c}{Simulation Parameter Settings} \\ \cmidrule(r){1-2}
\rowcolor{white}
Parameter & Default Value \\ \midrule
$\lambda$ & \begin{tabular}[x]{@{}l@{}}$1.0\mathrm{e}{-2}$\end{tabular}\\
$\theta$ & \begin{tabular}[x]{@{}l@{}}5\end{tabular}\\
$r$ & \begin{tabular}[x]{@{}l@{}}$10$\end{tabular}\\
$\beta$ & \begin{tabular}[x]{@{}l@{}}4.0\end{tabular}\\
$p$ & \begin{tabular}[x]{@{}l@{}}$0.2$\end{tabular}\\
 \bottomrule
\end{tabular}
    \caption{Table of default parameter values when held constant as part of the numerical simulation}
    \label{tab:simulation}
\end{table}

We examine EAoB and EAoC with different network parameter settings using numerical simulation. Unless stated otherwise, the default network parameter settings for simulation are provided in~\Cref{tab:simulation}. The EAoB and EAoC for each parameter settings is determined using Monte Carlo simulation, simulating $250000$ time slots per trial.

\begin{figure}
    \centering
    \includegraphics[width=1.1\linewidth]{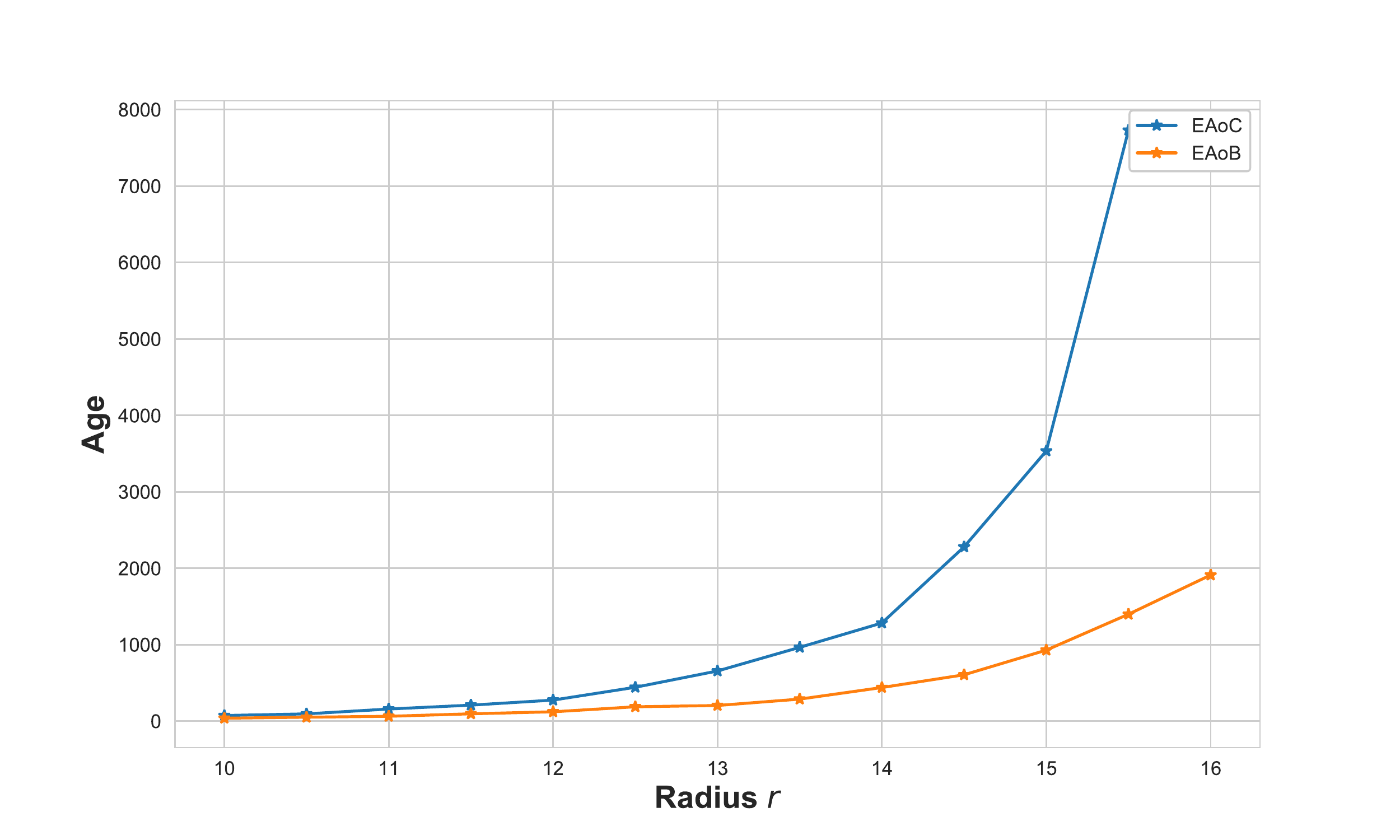}
    \caption{Instance-independent Age scaling with radius $r$ going from $10$ to $15.5$}
    \label{fig:age_r_scaling}
\end{figure}

\begin{figure}
    \centering
    \includegraphics[width=1.1\linewidth]{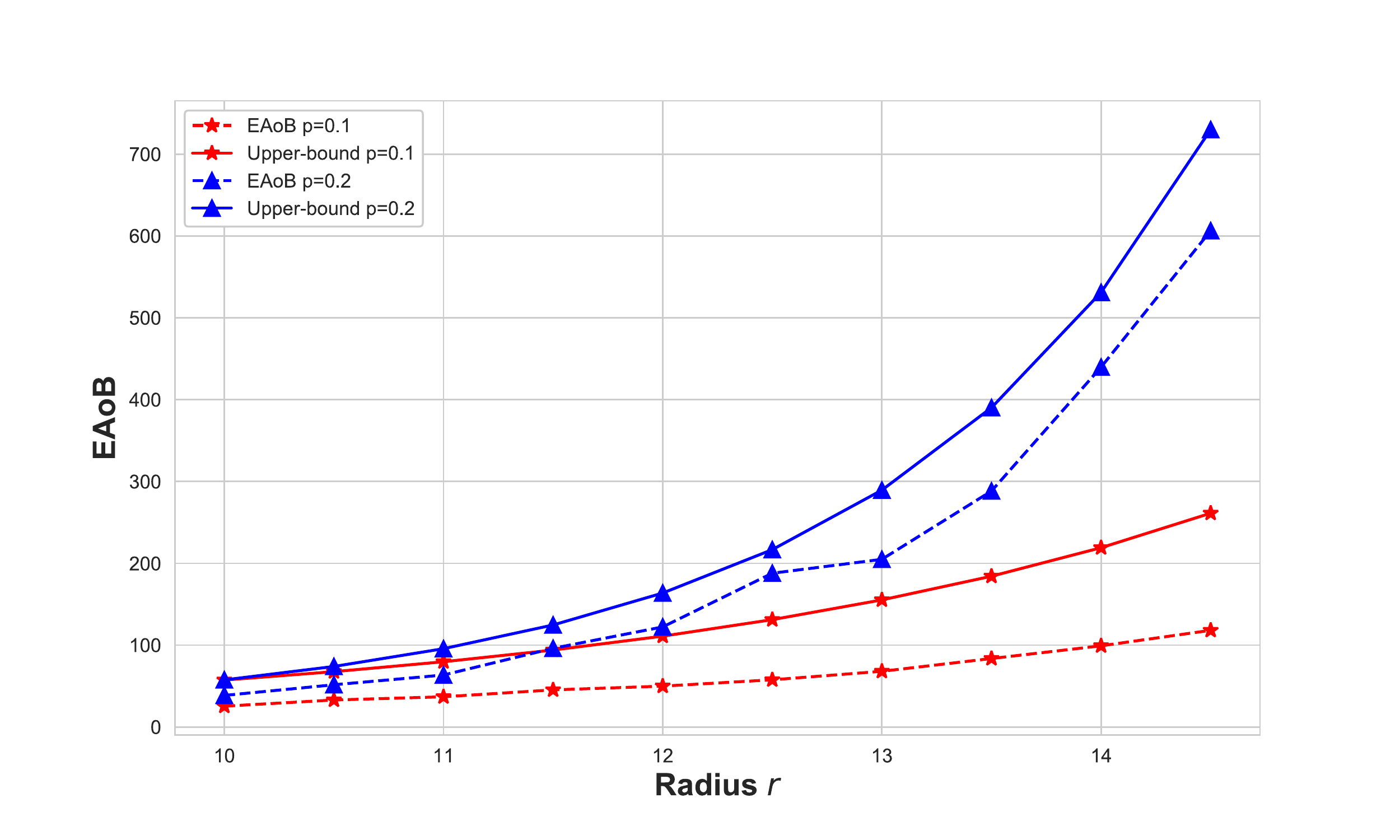}
    \caption{Instance-independent EAoB scaling with radius $r$ going from $10$ to $14$}
    \label{fig:aob_r_scaling}
\end{figure}
\begin{figure}
    \centering
    \includegraphics[width=1.1\linewidth]{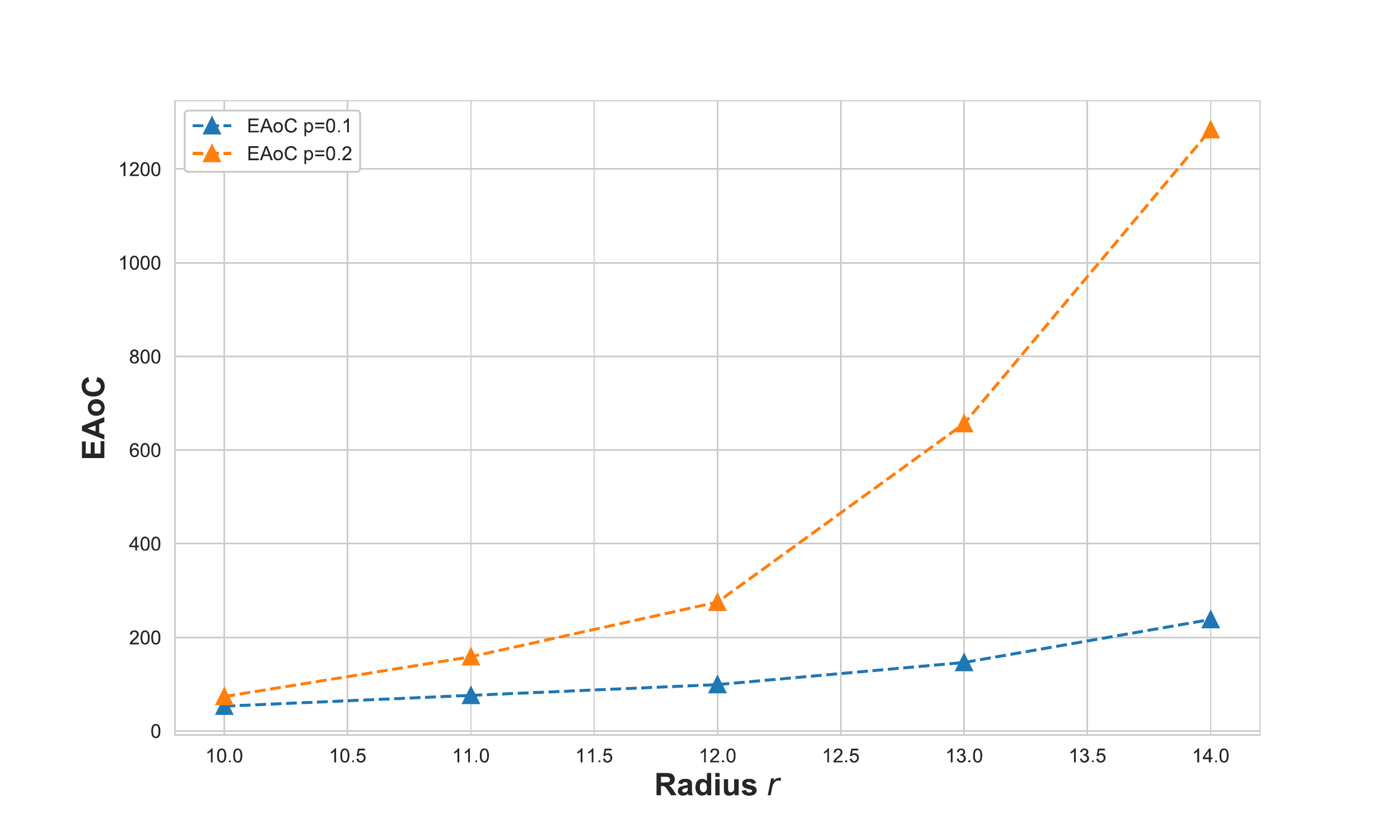}
    \caption{Instance-independent EAoC scaling with radius $r$ going from $10$ to $14$}
    \label{fig:aoc_r_scaling}
\end{figure}

\begin{figure}
    \centering
    \includegraphics[width=1.1\linewidth]{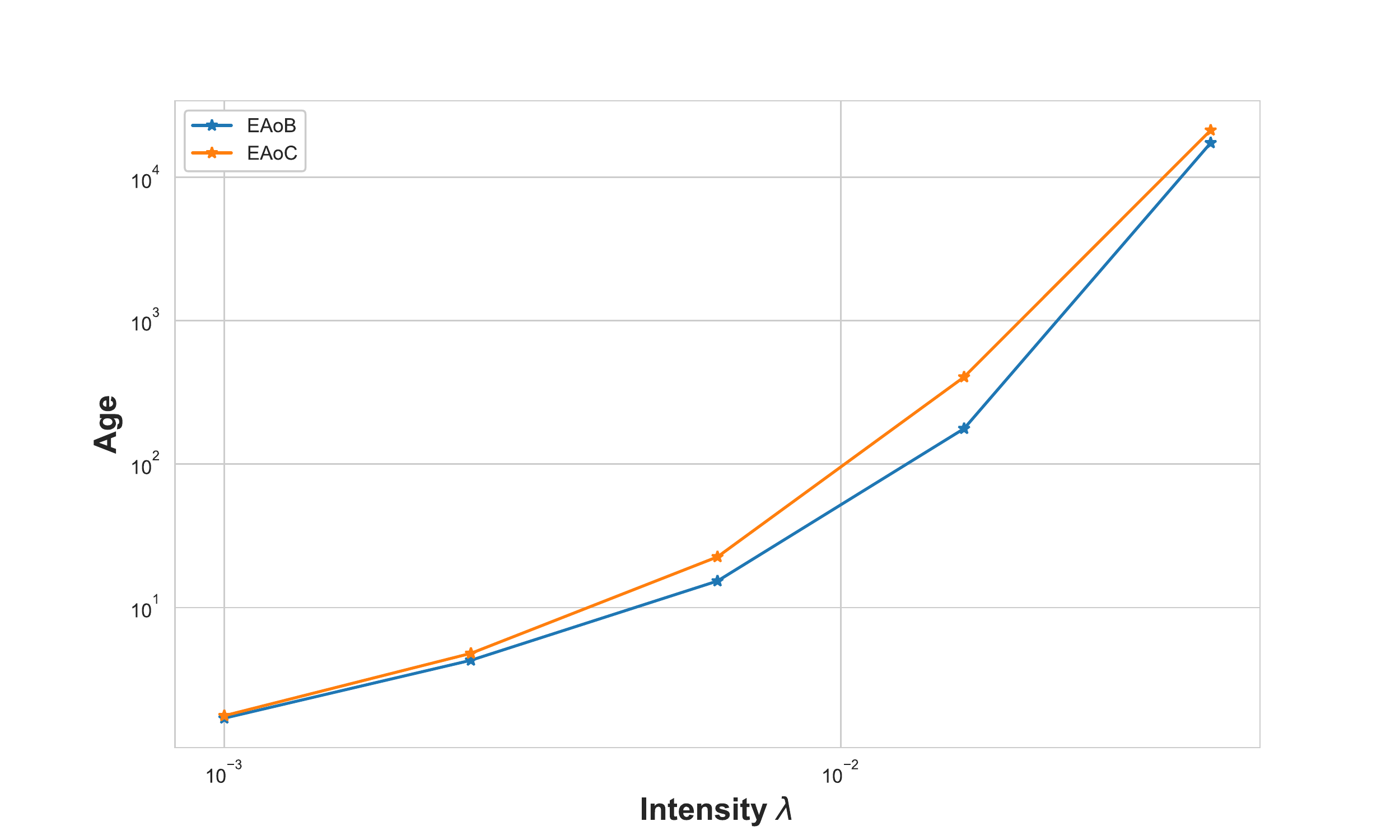}
    \caption{Instance-independent EAoB and EAoC scaling, with intensity $\lambda$ going from $1e-3$ to $1e-1$.}
    \label{fig:density_scaling}
\end{figure}
In~\Cref{fig:age_r_scaling}, the scaling behavior of EAoB and EAoC combined is shown as the radius of $b_2(\mathcal{O},r)$ is increased. The scaling behavior is super-exponential in both cases, with EAoC consistently larger than EAoB for the same radius. This is expected, since in the collection scenario at most one packet can be delivered to the base station in a given time slot, whereas when broadcasting, it is possible for multiple receivers to get a packet simultaneously. Moreover, the interference observed at the base station in the collection case is, in expectation, larger since both $\Phi_I$ and $\Phi_N$ are sources of interference when collecting.~\Cref{fig:aob_r_scaling,fig:aoc_r_scaling} plot EAoB and EAoC versus radius for different values of the medium access probability, showing that a greater value of $p$ result in larger age for these parameter settings. In~\Cref{fig:density_scaling}, the density is varied on a logarithmic scale. The figure depicts the exponential growth of EAoB and EAoC with respect to node and interferer intensity $\lambda$.

%% file: texfiles/conclusion.tex
\section{Conclusion}\label{sec:conclusion}
We defined AoB and AoC as information freshness metrics suitable for the cases of broadcast and collection, respectively, in spatially-distributed wireless networks. We characterized the expected AoB and AoC when the locations of nodes and interferers are known and unknown. When the locations are known and the packet transmission process is stationary, we showed that expected AoB and AoC were equivalent to the expected broadcast delay and collection delay, respectively. Upper-bounds were found in the instance-independent scenario: the AoB upper-bound is a solution to a differential equation, and the AoC upper-bound uses the solution to the worst-case packet delivery success probability given a small exclusion radius $\epsilon$. We demonstrated through numerical simulation the relation between AoB and AoC and network parameters such as density and medium access probability.

%% file: texfiles/appendices.tex
\appendices

\section{Proof of~\Cref{cl:br_delay}}\label{app:br_delay_pf}
\begin{proof}
For any time $t$ the packet reception vector $\Vec{{\mathbbm{1}}}^{\phi_I}_{\mathcal{O}}[t]=\left\{\mathbbm{1}^{\phi_I}_{\mathcal{O}i}\right\}_{i=1}^{|\phi_N|}$ can take values from a finite set of states $\mathcal{N}$ of cardinality $2^{|\phi_N|}$.  Therefore, the packet reception process over time is a discrete-time Markov Chain. Since the packet reception process is stationary and time-invariant by virtue of the ALOHA transmissions and i.i.d fading, this Markov Chain is irreducible, aperiodic, and positive recurrent. We denote the probability that the process is in state $u\in\mathcal{N}$ to be $p_u$.  WLOG, we denote the probability that the process transitions from state $j\in\mathcal{N}$ at time $t$ to $k\in\mathcal{N}$ at time $t+1$ to be $p_{jk}$ for any time $t\geq 0$. By Kolmogorov's criterion for Markov chain time-reversibility, the process is time-reversible if and only if the condition
\begin{align}\label{eq:kolmogorov}
    p_{{j_1}{j_2}}p_{{j_2}{j_3}}\hdots p_{{j_{n-1}}{j_n}}p_{j_{n}{j_1}}=p_{j_{1}{j_n}}p_{{j_{n}}{j_{n-1}}}\hdots p_{{j_3}{j_2}}p_{{j_2}{j_1}}
\end{align}
 is satisfied for all finite sequences of states $j_1,j_2,\hdots,j_n\in\mathcal{N}$.

Since the event the process is in state $j$ at time $t$ is independent of the event it is in state $k$ at time $t+1$ for all time $t\geq 0$, the condition in~\Cref{eq:kolmogorov} holds, establishing that the packet reception process is time-reversible. As shown  in~\Cref{fig:three graphs}, for any time $t$, the AoB $B^{\phi}_{\mathcal{O}}[t]$ is shown to be equivalent to $D^{\phi}_{\mathcal{O}}(t')$, the broadcast delay of the time-reversed process starting from the same time slot. Since time-reversibility implies that expectations with respect to forward time are the same as that over reverse time, coupled with the fact that  $\mathbb{E}\left[D^{\phi}_{\mathcal{O}}[t]\right]=\mathbb{E}\left[D^{\phi}_{\mathcal{O}}\right]$, we may conclude that 

\begin{align*}
    \mathbb{E}\left[B^{\phi}_{\mathcal{O}}\right]=\mathbb{E}\left[D^{\phi}_{\mathcal{O}}\right]\,.
\end{align*}
\end{proof}